\documentclass[12pt]{iopart}
\usepackage{graphicx}
\usepackage{dcolumn}
\usepackage{bm}

\usepackage{iopams}

\expandafter\let\csname equation*\endcsname\relax

\expandafter\let\csname endequation*\endcsname\relax

\usepackage{amsmath}
\usepackage{amssymb}
\usepackage{amsthm}
\usepackage[utf8]{inputenc}
\usepackage[english]{babel}
\usepackage{xcolor}

\newtheorem{theorem}{Theorem}
\newtheorem{lemma}{Lemma}
\newtheorem{definition}{Definition}
\newtheorem{corollary}[theorem]{Corollary}

\begin{document}

\title{Quantifying coherence of quantum measurements}

\author{Kyunghyun Baek$^1,^3$,, Adel Sohbi$^1$,
Jaehak Lee$^1$, Jaewan Kim$^1,^3$ and Hyunchul Nha$^2,^3$}
\address{$^1$ School of Computational Sciences, Korea Institute for Advanced Study, Seoul, 02455, Republic of Korea}
\address{$^2$ Department of Physics, Texas A$\&$M University at Qatar, Education City, P.O. Box 23874, Doha, Qatar}
\address{$^3$ Author to whom any correspondence should be addressed.}

\eads{\mailto{kbaek@kias.re.kr}, \mailto{ jaewan@kias.re.kr} , \mailto{ hyunchul.nha@qatar.tamu.edu} }

\begin{abstract}
In this work we investigate how to quantify the coherence of quantum measurements. First, we establish a resource theoretical framework to address the coherence of measurement and show that any statistical distance can be adopted to define a coherence monotone of measurement. For instance, the relative entropy fulfills all the required properties as a proper monotone. We specifically introduce a coherence monotone of measurement in terms of off-diagonal elements of Positive-Operator-Valued Measure (POVM) components. This quantification provides a lower bound on the robustness of measurement-coherence that has an operational meaning as the maximal advantage over all incoherent measurements in state discrimination tasks. Finally, we propose an experimental scheme to assess our quantification of measurement-coherence and demonstrate it by performing an experiment using a single qubit on IBM Q processor.
\end{abstract}

%
%
%
%
%

\section{Introduction}

With the development of quantum technologies, it has been widely perceived that quantum physics can offer enormous advantages in operational tasks. The so-called resource theoretical framework was introduced to systematically investigate which inherent features of quantum physics, e.g., entanglement \cite{Vedral1997,Horodecki2009}, contextuality \cite{Grudka2014} and non-Gaussianity \cite{Park2019,Zhuang2019,Lami2019,Albarelli2019,Takagi2019}, allow such advantages (see Ref. \cite{Chitambar2019} for more details). Particularly, quantum coherence is considered as one of the key ingredients that can offer quantum advantages in various forms. Its resource theoretical framework was initially developed for the quantification of coherence in quantum states \cite{Baumgratz2014}. This led to extensive investigations of its role in computation \cite{Hillery2016,Shi2017}, thermodynamics \cite{Korzekwa2016,Santos2019} and metrology \cite{Marvian2016,Giorda2017}. Furthermore, it was studied how closely quantum coherence is related to other fundamental notions such as entanglement \cite{Streltsov2015}, correlations \cite{Ma2016} and nonclassicality \cite{Tan2017} (see Ref. \cite{Streltsov2017} for more details).
These studies naturally stimulated interest on what advantages the coherence in quantum states may provide \cite{Yadin2016}. More recently, resource theoretical framework was established also for the quantification of coherent operations \cite{Theurer2019}.

Together with coherence in quantum states and its manipulation, coherence in quantum measurement must also be considered as a resource since the quantum nature of measurement is required to access the coherence of quantum states in experiment. 
 If a measurement cannot address coherence of quantum states, its resulting measurement statistics does not provide any information on coherence under investigation \cite{Theurer2019}.
Furthermore, it is essential to understand the characteristics of a measuring device to analyze measurement in various scenarios such as the state-independent contextuality \cite{Cabello2008,Badziag2009,Yu2012}, quantum measurement engine \cite{Elouard2017, Elouard2017-2, Elouard2018} and measurement-based quantum computation \cite{Raussendorf2001,Raussendorf2003,Briegel2009,Briegel2009-2}. 
It is thus of crucial fundamental importance to study a rigorous quantification of coherence in quantum measurement and its possible experimental characterization. 
In a related context, the resource theoretical approach for coherent operations was employed in \cite{Xu2019}. It considered a quantum measurement as an operation that maps a quantum state to a statistical distribution according to the Gleason's theorem \cite{Gleason1957}, which was also demonstrated experimentally via detector tomography.  
On the other hand, the robustness to noise was adopted to quantify coherence of measurement \cite{Oszmaniec2019}, which is in line with the operational characterization of general convex resource theories \cite{Oszmaniec2019,Takagi2019-2,Uola2019}. 
However, these quantifications require convex optimization that can be challenging  for higher dimensions. From an experimental point of view, an operational assessment of coherence of measurement was also provided in \cite{Cimini2019} apart from the resource theoretical framework.


In this paper, we establish a resource theoretical framework for the quantification of coherence in quantum measurements and introduce proper monotones by using statistical distance measures. In addition, we introduce a readily computable coherence monotone of measurement which takes into account the off-diagonal elements of each POVM component. We show that this monotone gives a lower bound on the general robustness that is closely related to the maximal advantage over all incoherent measurements in the state discrimination task \cite{Oszmaniec2019}. 
Finally, to experimentally assess the coherence of quantum measurement, we apply quantum process tomography in \cite{Isaac1997} to detector tomography by which we can straightforwardly obtain the off-diagonal elements from its statistical distributions. The key benefit of this approach may be its simplicity and directness as it avoids a computationally demanding post-processing that becomes challenging for high dimensional systems.

\section{ Preliminaries}

Coherence in a quantum state essentially refers to the quantum superposition principle. Whether a quantum state has a superposition nature or not depends on the choice of the basis states to represent the given state. 
Here we confine ourselves to the cases where an orthonormal basis $\{|i\rangle\}_{i=0}^{d-1}$ is specified in $d$-dimensional Hilbert space $\mathcal H_d$ such as computational basis and energy basis. We call this {\it an incoherent basis}.

In the resource theory of coherence, we say that a state is free, i.e., {\it incoherent}, if it corresponds to a statistical mixture of the incoherent basis states $\{|i\rangle\}_{i=0}^{d-1}$. Hence, a state $\rho$ is incoherent if and only if it has the form of
\begin{align}
  \rho=\sum_{i =0}^{d-1} p_i |i\rangle\langle i|.
\end{align}
We define the total dephasing operation $\Phi$, which completely destroys coherence of states,
\begin{align}
  \Phi(\rho)=\sum_{i=0}^{d-1} \langle i|\rho|i\rangle |i\rangle\langle i|.
\end{align}

\subsection {Classification of quantum measurements}

A  quantum measurement on $\mathcal H_d$ is generally described by a Positive-Operator-Valued Measure (POVM), which is a set of positive operators, i.e., $\mathbf A=\{A_a\}_{a=0}^{n-1}$, satisfying the completeness relation $\sum_{a=0}^{n-1} A_a=I_d$ with $n$ the number of measurement outcomes. 
All POVMs on $\mathcal H_d$ with $n$ outcomes  form the convex set $\mathcal M(d,n)$ as we define the convex combination of  $\mathbf A,\mathbf B \in \mathcal M(d,n)$ by $p \mathbf A+(1-p)\mathbf B=\{pA_a+(1-p)B_a\}_{a=0}^{n-1}$ with $0\leq p\leq 1.$
We say that a POVM is free if a probability distribution of measurement outcomes is independent of coherence of quantum states. This is defined formally as follows  \cite{Theurer2018}.
	A POVM $\mathbf A\in \mathcal M(d,n)$ given by $\{A_a\}_{a=0}^{n-1}$ is free, i.e., an {\it incoherent measurement} (IM) if and only if 
\begin{align}
\Tr[A_a \Phi(\rho)]=\Tr[A_a\rho]
\end{align}
for all states $\rho$ and all $a\in \{0,...,n-1\}$. Equivalently, it is free if and only if all POVM components are written in an incoherent form
\begin{align}\label{IncM}
A_a=\sum_{i=0}^{d-1} \alpha_{i|a} |i\rangle\langle i|
\end{align}
with $\alpha_{i|a}=\langle i|A_a|i\rangle$.

\subsection {Classification of quantum operations}

We can generally deal with a quantum operation  $\mathcal E$ in the framework of completely positive trace preserving maps defined by a set of Kraus operators $\{K_\mu\}$. Its action on a state $\rho$ is expressed as $$\mathcal E(\rho)=\sum_{\mu=0}^{n_{\mathcal E}-1} K_\mu \rho K_\mu^\dagger,$$
where $\sum_\mu K_\mu^\dagger K_\mu = I$ with $n_{\mathcal E}$ the total number of Kraus operators.
$\mathcal E$ is a {\it maximally incoherent operation} (MIO) if it maps the set of incoherent states $\mathcal{I}$ to its subset, i.e., $\mathcal E(\mathcal I) \subset \mathcal I$.
However, MIO does not necessarily imply that an output state associated with a selective measurement outcome is incoherent. 
Namely, there is a case $\rho\in \mathcal I \rightarrow \mathcal E(\rho)=\sum_{\mu} K_\mu \rho K_\mu^\dagger\ \in \mathcal I$, but $K_\mu \rho K_\mu^\dagger/p_\mu \notin \mathcal I$  for a specific $\mu$ with $p_\mu=\text{Tr}[K_\mu \rho K_\mu^\dagger]$.
If all $K_\mu$ associated with $\mathcal E$ map incoherent states to incoherent states, then it is called {\it an incoherent operation} (IO) \cite{Baumgratz2014}, i.e., $K_\mu \rho K_\mu^\dagger /p_\mu \in \mathcal I$ for all $\mu$.

A quantum operation with $\{K_\mu\}$ acting on a quantum state $\rho$ can be equivalently addressed as acting on quantum measurements $\mathbf A$ as
$$\text{Tr}[K_\mu\rho K^\dagger_\mu A_a]=\text{Tr}[\rho K^\dagger_\mu A_a K_\mu].$$
In view of coherence in states, we defined IOs mapping incoherent states to themselves. However, an IO can generate coherent measurements from incoherent ones. For instance, let us consider an IO described by the following Kraus operators
$$K_0=|0\rangle\langle +| \text{ and }K_1=|1\rangle\langle -|,$$
where $|\pm\rangle=(|0\rangle\pm|1\rangle)/\sqrt{2}$ in the incoherent basis $\{|0\rangle,|1\rangle\}$. Conditional states associated with outcomes 0 and 1 are $|0\rangle$ and $|1\rangle$ with probabilities $p_0=\langle +|\rho|+\rangle$ and $p_1=\langle -|\rho|-\rangle$, respectively. However, we obtain a coherent measurement by applying its dual operation to an incoherent measurement $\{A_a\}$ such that
\begin{align*}
K_0^\dagger A_{a} K_0&=\langle 0|A_{a}|0\rangle |+\rangle\langle +|,\\
K_1^\dagger A_{a} K_1&=\langle 1|A_{a}|1\rangle |-\rangle\langle -|.
\end{align*}

Therefore we need a stricter definition of IO that cannot generate coherence both from incoherent states and from incoherent measurements. In this context, both $K_\mu$ and $K_\mu^\dagger$ have to be incoherent.
This additional condition elevates IO to {\it strictly incoherent operation} (SIO) \cite{Yadin2016,Winter2016}, which is a set of operations that can neither create nor detect coherence in an input state. 
One can concisely express a Kraus operator of SIO as \cite{Yadin2016,Chitambar2016b}
\begin{align}\label{SIO}
  K_\mu=\sum_{i=0}^{d-1} c_{\mu,i} |\pi_\mu(i)\rangle\langle i|=V_\mu \tilde{K}_\mu,
\end{align}
where $V_\mu$ is a unitary operation corresponding to  a permutation $\pi_\mu$ and $\tilde K_\mu=\sum_{i=0}^{d-1} c_{\mu,i} |i\rangle\langle i|$ is a genuinely incoherent Kraus operator  \cite{Yadin2016}. The completeness relation $\sum_\mu K_\mu^\dagger K_\mu=I_d$ implies $\sum_\mu |c_{\mu,i}|^2=1$ for each $i$. Thus, one can rewrite it as $c_{\mu,i}=e^{i \theta_{\mu,i}}\sqrt{p_{\mu|i}}$ in terms of phase and conditional probability of $\mu$ given $i$. An operational interpretation of SIO was found in \cite{Yadin2016} in terms of interferometry. We remark that Ref. \cite{Theurer2019,Xu2019} studied coherence of quantum operations, considering a general set of free operations defined by detection incoherent operations.

\section {Quantification of coherence in quantum measurements}

Based on the axiomatic quantification of coherence of states \cite{Baumgratz2014}, we introduce a list of properties that a legitimate coherence monotone $\mathcal C$ of a measurement $\mathbf A\in \mathcal M(d,n)$ must fulfill as follows.
\begin{itemize}
  \item[($\mathcal C$1)] Faithfulness: $\mathcal C(\mathbf A)\geq 0$ for all POVMs with equality if and only if $A$ is incoherent.
  \item[($\mathcal C$2)] Monotonicity: $\mathcal C$ does not increase under the dual operation of any nonselective SIO $\mathcal{E}$, i.e., $\mathcal{C}(\mathcal E^* (\mathbf A))\leq\mathcal{C}(\mathbf A)$, where $\mathcal E^*(\mathbf A)$ denotes a measurement $\{\sum_\mu K^\dagger_\mu A_a K_\mu\}_a$ given after applying the nonselective dual SIO to $\mathbf A$.
  \item[($\mathcal C$3)] Strong monotonicity: $\mathcal C$ does not increase under the dual operation of any selective SIO $\{K_\mu\}$, i.e., $\mathcal{C}(\mathbf A')\leq\mathcal{C}(\mathbf A)$, where the expanded POVM $\mathbf A'=\{K^\dagger_\mu A_a K_\mu\}_{a,\mu}$ is given after applying the SIO to $\mathbf A$ (See Appendix \ref{Appdx1} for more detailed discussion on $\mathbf A'$).
  \item[($\mathcal C$4)] Convexity: For $n$-outcome POVMs given by $\mathbf A_k=\{A_{a|k}\}_{a=0}^{n-1}$, $\mathcal C$ is a convex function of the measurement, i.e., $\sum_k q_k \mathcal{C}(\mathbf{A}_k)\geq \mathcal{C}(\sum_k q_k \mathbf A_k)$, where $\sum_k q_k \mathbf A_k$ describes a POVM given by $\{\sum_k q_k A_{a|k}\}_{a=0}^{n-1}$. This measurement is constructed by performing the $k$th measurement $\mathbf A_k$ with probability $q_k$ and combining each outcome.
\end{itemize}
Here, ($\mathcal C$2) (($\mathcal C$3)) imposes that a monotone should not increase by action of a (selective) SIO. Also, ($\mathcal C$4) guarantees that the loss of information about choice of measurement cannot increase the average coherence of measurements. Unlike the case of coherence of states, conditions ($\mathcal C$3) and ($\mathcal C$4) are not sufficient for the satisfaction of ($\mathcal C$2).

\subsection{ Statistical distance-based coherence monotones of measurement}

{ We are now ready to introduce coherence monotones of measurement.} To begin with, we consider a statistical distance from which we can define a coherence monotone as follows.
\begin{definition}
	A coherence monotone of measurement $\mathbf A$ is defined as
	\begin{eqnarray}\label{CohMonotone}
	\mathcal C_D (\mathbf A)=\min_{\mathbf M\in \mathcal I(d,n)} \sup_\rho D_\rho(\mathbf A,\mathbf M),
	\end{eqnarray}
	where $D_\rho(\mathbf A,\mathbf M)$ is a statistical distance between probability distributions $p_{\mathbf A}(a)=\Tr[\rho  A_a]$ and $p_{\mathbf M}(a)=\Tr[\rho  M_a]$. Here $\mathbf M$ belongs to the set $\mathcal I(d,n)$ of incoherent measurements defined in $d$ dimension with $n$ outcomes.
\end{definition}
A proper statistical distance yields $D_\rho(\mathbf A,\mathbf M)=0$ if and only if both distributions are identical. This property implies the satisfaction of ($\mathcal C$1) for $\mathcal C_D$. In addition, $\mathcal C_D$ satisfies ($\mathcal C$2) due to its definition, and ($\mathcal C$4) as long as the statistical distance is convex. Detailed mathematical proofs are given in  \ref{Appdx2}.

{
In the state-based approach to quantum coherence, the relative entropy leads to a coherence measure that has  important operational meanings such as coherence distillation \cite{Winter2016}. It is further adopted to address the quantumness of measurement in quantification of state coherence based on POVMs \cite{Bischof2019}. Similarly, we take the relative entropy as a statistical distance as follows. 
}
\begin{definition}
Relative entropy-based coherence monotone of measurement is defined as
\begin{eqnarray}\label{CohMeasureD}
  \mathcal C_S(\mathbf A)=\min_{\mathbf M\in \mathcal I(d,n)} S(\mathbf A\|\mathbf M),
\end{eqnarray}
with the channel divergence \cite{Cooney2016,Leditzky2018,Gour2019} $S(\mathbf A\|\mathbf M)=\sup_\rho H_\rho(\mathbf A\|\mathbf M)$,
where $H_\rho(\mathbf A\|\mathbf M)=\sum_{a=0}^{n-1} p_{\mathbf A}(a)\log [ p_{\mathbf A}(a)/p_{\mathbf M}(a)]$ is the relative entropy between probability distributions $p_{\mathbf A}(a)=\Tr[\rho  A_a]$ and $p_{\mathbf M}(a)=\Tr[\rho M_a]$.
\end{definition}

{ Consequently, this coherence monotone satisfies the strong monotonicity with its proof provided in \ref{Appdx2}. It is worth noting that this monotone coincides with the relative entropy of a dynamical resource introduced in \cite{Gour2019-2}, as we regard a quantum measurement as a channel mapping a quantum state to a statistical distribution. }Similarly, one can take other monotones suggested in \cite{Saxena2019, Theurer2019} to quantify coherence of quantum measurements. 
\newline


\subsection{ Robustness and intuitive quantification of coherence}

On the other hand, one can employ the robustness to quantify the coherence of measurement as the minimal amount of mixing with a measurement that makes the given measurement incoherent. Namely, the robustness is defined as
\begin{align}\label{Rob}
  \mathcal R_C(\mathbf  A)=\min \left\{s \bigg| \frac{\mathbf A+s\mathbf M}{1+s} \in \mathcal I(d,n) \right\},
\end{align}
where the minimization is performed over all measurements $\mathbf M\in \mathcal M(d,n)$. 
Indeed, the robustness was introduced for the mathematical quantification of entanglement in \cite{Vidal1999}.
Recently, in general convex resource theories \cite{Takagi2019-2}, it was shown that the robustness allows an operational interpretation in some discrimination tasks whenever free resources form convex subset of resource objects.  Particularly, in the resource theories of quantum measurements, the robustness indicates the maximum advantage over all free measurements in state discrimination tasks \cite{Oszmaniec2019, Uola2019}.


To evade the convex optimization procedure, {however,} it can be useful to establish a more intuitive quantification of coherence. As  the $l_1$-norm of coherence of states was introduced in \cite{Baumgratz2014}, an intuitive quantification would be accomplished by taking into account off-diagonal elements of a considered POVM. To this end, we introduce a $d$ by $d$ matrix 
\begin{align}\label{Omega}
\Omega(\mathbf A) =  \sum_{i,j=0}^{d-1} \left( \sum_{a=0}^{n-1} |\langle i|A_a|j\rangle|\right) |i\rangle\langle j|
\end{align}
that allows us to look into all off-diagonal elements of a POVM $\mathbf A$, where $|\cdot|$ denotes the absolute value.
Then, its diagonal elements are all unity due to the completeness relation, and it is reduced to the identity matrix $I_d$ if and only if $\mathbf A$ is  an incoherent measurement. From this property, we define a readily computable coherence monotone of measurement as follows.
\begin{definition}
A $l_\infty$ norm-based coherence monotone of measurement is defined  as
\begin{align}\label{CohMeasure}
  \mathcal C_{l_\infty}(\mathbf A)&=\min_{\mathbf M\in \mathcal I(d,n)} \| \Omega(\mathbf A)- \Omega(\mathbf M)\|_\infty\\
  &= \max_{i< j} \sum_{a=0}^{n-1} |\langle i |A_a |j \rangle|.\nonumber
\end{align}
where the $l_\infty$ matrix norm gives the largest absolute value among each element of a matrix, i.e., $\|\cdot\|_\infty=\max_{i,j}|\langle i|\cdot|j\rangle|$.
\end{definition}
\noindent As desired, this coherence monotone fulfills all requirements ($\mathcal C$1-4) as proved in Appendix \ref{Appdx3}. 

One may wonder if it can be possible to quantify coherence of measurement by employing the $l_1$ matrix norm
such as 
$$\mathcal C_{l_1}(\mathbf A)=\min_{\mathbf M\in \mathcal I(d,n)} \| \Omega(\mathbf A)- \Omega(\mathbf M)\|_1=  \sum_{a=0}^{n-1}\sum_{i\neq j} |\langle i |A_a|j\rangle|.$$ 
However, it is not difficult to find counter examples showing that $\mathcal C_{l_1}(\mathbf A)$ increases after applying the dual of a SIO (see  \ref{Appdx4}). Nevertheless, it is worth considering $\mathcal C_{l_1}$ together with $\mathcal C_{l_\infty}$ when we estimate the robustness from the following relations.
\begin{theorem}
For any POVM $\mathbf A \in \mathcal M(d,n)$, the following holds
\begin{align}\label{Relation}
\mathcal C_{l_\infty}(\mathbf  A)\leq \mathcal R_C(\mathbf A)\leq \frac{1}{2} \mathcal C_{l_1}(\mathbf A).
\end{align}
\end{theorem}

The proof is provided in  \ref{Prf of Thm1}. The relation allows us to bound the robustness measure from below based on the absolute value of the off-diagonal elements. Particularly, for 2-dimensional cases, they become equality as follows.

\begin{corollary}
For any POVM $\mathbf A \in \mathcal M(2,n)$, it holds that
\begin{align}\label{Relation}
\mathcal C_{l_\infty}(\mathbf  A)= \mathcal R_C(\mathbf A)=\frac{1}{2}\mathcal C_{l_1}(\mathbf A).
\end{align}
\end{corollary}
\noindent It is because the lower and the upper bound become identical by the definition of $\mathcal C_{l_\infty}$ and { $\mathcal C_{l_1}/2$. Here, the factor 1/2 appears because off-diagonal terms are counted twice for $i> j$ and $i< j$ in $\mathcal C_{l_1}$, while they are counted only once for $i<j$ in $\mathcal C_{l_\infty}$}.

For higher dimensions $d>2$, {  their relationships are more involved.
To see their behaviors, we consider a dichotomic POVM $\mathbf G$ in a 3-dimensional system as an example, i.e., $\mathbf G \in \mathcal M(3,2)$. We further assume the amplitude damping channel is applied to $\mathbf G$, which is known as a SIO (see \ref{Ex} for more details). As a result, figure \ref{Comparing} shows that the equalities in \eqref{Relation} do not hold. It is because $\mathcal C_{l_\infty}$ takes only the maximal off-diagonal elements into account, while others take all off-diagonal elements. We expect that the difference among them usually become more substantial for higher dimensions.  Nevertheless, there can be POVMs satisfying $\mathcal C_{l_\infty}(\mathbf A)=\mathcal C_{l_1}(\mathbf A)/2$ for higher dimensions, if off-diagonal elements of $\Omega(\mathbf A)$ corresponding to specific incoherent states $|i\rangle$ and $|j\rangle$ are non-zero while other terms vanish.
Additionally, we remark that $\mathcal R_{\mathcal C}$ and $\mathcal C_{l_\infty}$ monotonically decrease with the damping rate, while the upper bound $\mathcal C_{l_1} $ does not. It is because $\mathcal C_{l_1} $ is not a monotone as mentioned above. }  

\begin{figure*}[t]
  \centering
    \includegraphics[width =8cm]{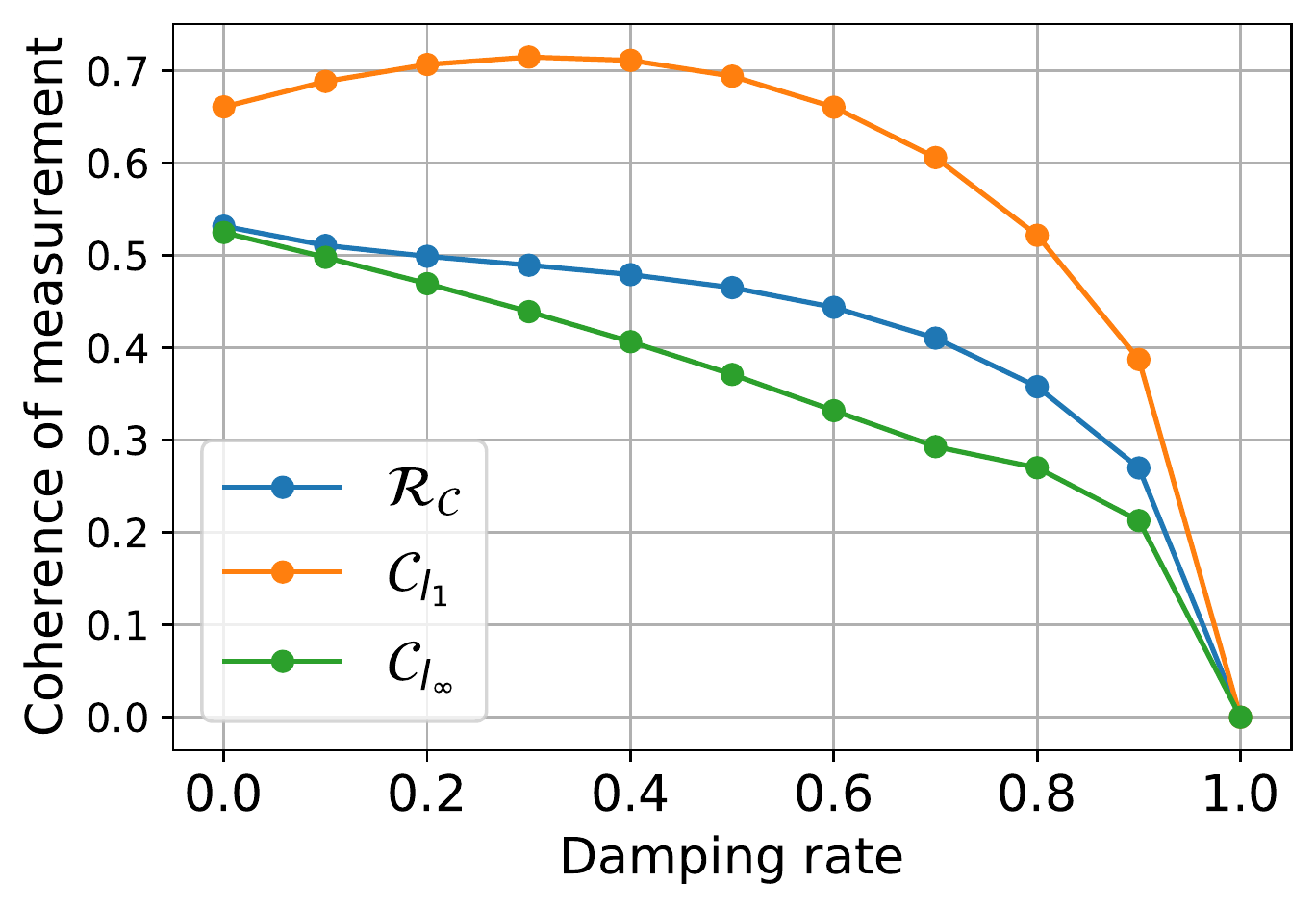}
    \caption{$\mathcal R_{\mathcal C}$, $\mathcal C_{l_1}$ and $\mathcal C_{l_\infty}$ for a dichotomic POVM $\mathbf G$  in a 3-dimensional system $\mathbf G \in \mathcal M(3,2)$ to which the amplitude damping channel is applied. Detailed POVM components are given in \ref{Ex}.
}\label{Comparing}
\end{figure*}



\section{Experimental scheme to assess coherence of measurement}

General procedures for detector tomography consist of two steps: {\it (i) data collection} as preparing a set of probe states and measuring each of them by an unknown measurement and {\it (ii) data analysis}, the so-called global reconstruction as finding the optimal physical POVM consistent with the data \cite{Lundeen2009}.  However, the global reconstruction can be challenging for high dimensional systems. We instead employ a procedure introduced in \cite{Isaac1997} for quantum process tomography. For an unknown measurement $\mathbf A=\{A_a\}_{a=0}^{n-1}$ on $\mathcal H_d$, a set of $d^2$ linearly independent probe states $\{\rho_k\}_{k=1}^{d^2}$ is required to explicitly obtain $\mathbf A$. We denote by $p(a|k)$ the probability of obtaining outcome $a$ associated with the POVM $\mathbf A$ for a probe state $\rho_k$.
Without loss of generality, we write them in an incoherent basis as $A_a=\sum_{i,j=0}^{d-1} \alpha_{ij|a}|i\rangle\langle j|$ and $\rho_k=\sum_{i,j=0}^{d-1}\beta_{ij|k} |i\rangle\langle j|$, where $\alpha_{ij|a}=\langle i|A_a|j\rangle$ and $\beta_{ij|k}= \langle i |\rho_k|j\rangle$, respectively. 
Each probability is written in terms of $\alpha_{ij|a}$ and $\beta_{ij|k}$ as $p(a|k)=\Tr[\rho_k A_a]=\sum_{i,j=0}^{d-1}\beta_{ij|k}\alpha_{ji|a}$ for all $a$ and $k$.

We introduce $d^2$-dimensional real column vectors $\bm{\gamma}_k$ and $\bm{\chi}_a$ that contain information on $\rho_k$ and $A_a$, respectively. That is, we can rewrite $p(a|k)$ as 
\begin{align}
p(a|k)={\bm{\gamma}}_k^T \bm{\chi}_a,
\end{align}
where $\bm{\chi}_a$ and $\bm\gamma_k$ are defined as
\begin{align*}
(\bm\chi_a)_{x}= 
\begin{cases}
\alpha_{qq|a} & \text{for } q = r \\
\sqrt{2}\text{Re}[\alpha_{qr|a}] &\text{for } q < r\\
\sqrt{2}\text{Im}[\alpha_{rq|a}]  &\text{for } q > r
\end{cases},
\\
(\bm\gamma_k)_{x}= 
\begin{cases}
\beta_{qq|k} & \text{for } q = r \\
\sqrt{2}\text{Re}[\beta_{qr|k}] &\text{for } q < r\\
\sqrt{2}\text{Im}[\beta_{rq|k}]  &\text{for } q > r
\end{cases}
\end{align*}
with integers $0 \leq x \leq d^2-1$ and $0\leq q,r\leq d-1$ holding $x=qd+r$. For instance, we have $\bm \gamma_k = (\beta_{00|k}, \sqrt{2}\text{Re}[\beta_{01|k}],\sqrt{2}\text{Im}[\beta_{01|k}],\beta_{11|k})^T$ for $d=2$. Then, as we define $d^2 \times d^2$ square matrix as $\bm\Gamma=(\bm\gamma_1,\bm\gamma_2,...,\bm\gamma_{d^2})$ and a probability vector for $a$ as $\bm\mu_a=(p(a|1),...,p(a|d^2))^T$, all measured statistics related to an outcome $a$ is written in terms of these vectors as
\begin{align}\label{DetectorT}
\bm\Gamma^T\bm \chi_a=\bm\mu_a.
\end{align}
Linear independence of $\rho_k$ implies that the inverse of $\bm\Gamma$ exists. Therefore, by applying its inverse $(\bm\Gamma^T)^{-1}$ on both sides, one can obtain the unknown vector $\bm\chi_a$ that explicitly determines the POVM component $A_a$.

\begin{figure*}[t]
  \centering
    \includegraphics[width =15cm]{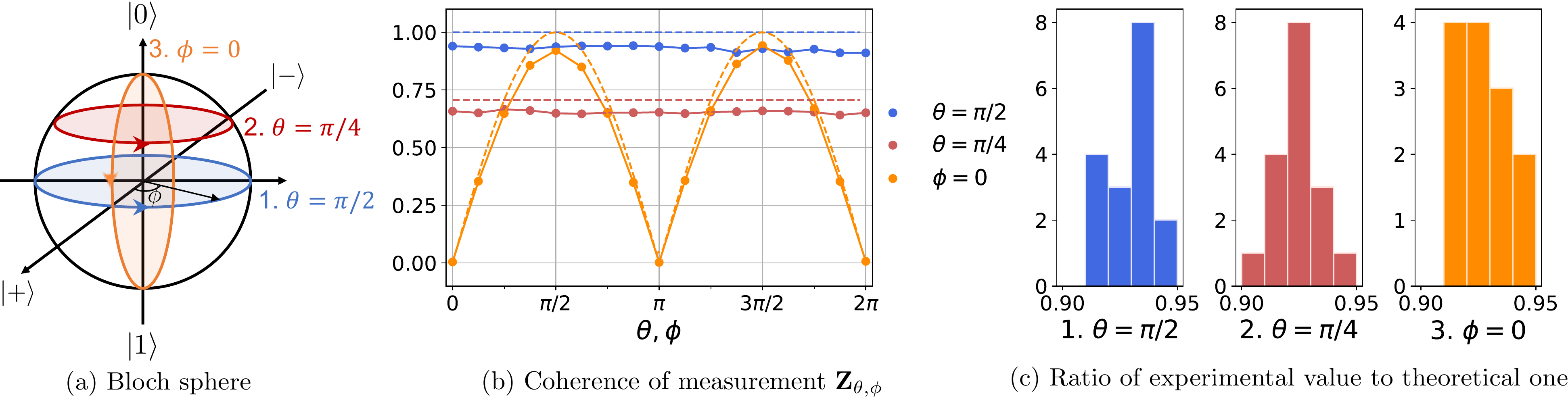}
    
    \caption{Demonstration for single qubit measurements on IBM Quantum systems. (a) Each point in the Bloch sphere represents the direction of measurement $\mathbf Z_{\theta,\phi}$ varying along the different paths. (b) Coherence of measurements $\mathbf Z_{\theta,\phi}$ along paths 1 (blue), 2 (red) and 3 (orange). Dashed curves show theoretical expectation values.  (c) Histograms of the ratio of an experimental value to the theoretical one. 
}\label{Figure1}
\end{figure*}

For our goal to experimentally assess the coherence of measurements, we consider the following family of linearly independent states,
\begin{align}\label{States}
|\psi_{kl}\rangle= 
\begin{cases}
|k\rangle & \text{for } k=l\\
({|k\rangle+|l\rangle })/{\sqrt2} &\text{for } k>l\\
({|k\rangle+i|l\rangle})/{\sqrt2} &\text{for } k<l
\end{cases},
\end{align}
with $k,l=0,\cdots,d-1$. 
Preparing this family of states and measuring them by an unknown POVM $\mathbf A$, one can directly construct the components of a POVM from the relations
\begin{align*}
p(a|kl)-\frac{p(a|kk)+p(a|ll)}{2}=
\begin{cases}
\text{Re}[\langle k|A_a|l\rangle] & \text{ for } k>l\\
\text{Im}[\langle l|A_a| k \rangle]&\text{ for }  k<l
\end{cases},
\end{align*}
where we denote a probability to obtain an outcome $a$ for a prepared state $|\psi_{kl}\rangle$ by $p(a|kl)=\langle \psi_{kl}|A_a|\psi_{kl}\rangle$.  
The benefit of this method is that one can avoid the reconstruction of whole POVM components that may be challenging for high dimensional systems. In addition, it is possible to obtain its off-diagonal elements, $\langle k |A_a|l\rangle$ straightforwardly, and also to measure it selectively by preparing $|\psi_{kl}\rangle$ for specific values of $k,l$. Here, the quantum superposition of $|k\rangle$ and $|l\rangle$ serves as a resource to assess the amount of coherence in that basis.

We demonstrate our proposed scheme in a single qubit experiment on the 20 qubit IBM Quantum system ``Singapore''. In the IBM Q processor, qubit states are measured in computational basis, i.e., $\mathbf Z=\{|0\rangle,|1\rangle\}$. We manipulate measurement basis by applying the following single-qubit gate before measuring in  $\mathbf Z,$
\[V_{\theta,\phi}=
\begin{pmatrix}
\cos(\theta/2)&& e^{-i\phi}\sin(\theta/2)\\
-\sin(\theta/2) && e^{-i\phi} \cos(\theta/2)
\end{pmatrix},
\]
with the polar angle $\theta \in [0,\pi]$ and the azimuthal angle $\phi \in [0,2\pi]$.
As a result, it is equivalent to execute the measurement $\mathbf Z_{\theta,\phi}=\{V_{\theta,\phi}^\dagger |0\rangle, V_{\theta,\phi}^\dagger |1\rangle\}$. To assess its coherence of measurement, we prepare the family of states $\{|\psi_{kl}\rangle\}_{k,l=0,1}$ as specified in \eqref{States}. In the IBM Q processor, they are prepared by applying appropriate single-qubit gates to the initial state $|0\rangle$ (see \ref{App_Exp} for more details).

We consider three cases $(i)\; \theta=\pi/2$, $(ii) \;\theta=\pi/4$ and $(iii)\;\phi=0$ as illustrated in figure \ref{Figure1}-(a). For the cases $(i)\; \theta=\pi/2$ and $(ii) \;\theta=\pi/4$, we vary $\phi\in [0,2\pi]$ with an interval of $\pi/8$. Similarly, for the case $(iii)\;\phi=0$, we vary $\theta\in [0,2\pi]$ with an interval of $\pi/8$ . For the single qubit measurement $\mathbf Z_{\theta,\phi}$, its coherence of measurement is given by $\mathcal R(\mathbf Z_{\theta,\phi})=\mathcal C_{l_\infty}(\mathbf Z_{\theta,\phi})=\mathcal C_{l_1}(\mathbf Z_{\theta,\phi})/2=|\sin\theta|$ according to \eqref{Relation}.  We plot its theoretical and experimental values in figure \ref{Figure1}-(b). Here, each point denotes the average value of 10 runs with the setting of 8192 shots for each run on the IBM Q processor. For each point, the order of its standard deviation is less than $10^{-3}$, so that its error bar is smaller than the size of a point. 

In figure \ref{Figure1}-(b), the gap between experimental and theoretical values seems to be bigger with a larger coherence of measurement.  In order to see this trend more clearly, we draw histograms for the ratio of each experimental value to its theoretical one except for singular points where theoretical values become zero, i.e., $\theta=0,\pi,2\pi$ for $\phi=0$. The corresponding histograms in figure \ref{Figure1}-(c) show that the ratios are rather uniform in the range of 0.9 $\sim$ 0.95 insensitive to the degree of coherence. 

{ In practice, our experimental results were obtained by executing designed circuits in IBM Q processor and the prepared states would not be pure. However, we may assign the error or noise occurring to these probe states to the POVM under test. That is, we assume a pure probe state for analysis in Eq. (17) (Eq. (14) in the revised version), and importantly, this procedure will only underestimate the coherence of measurement, thereby providing a reliable lower bound for robustness, if the noise occurring is SIO (strictly incoherent operation): When we introduced the properties of our monotone of coherence, we proved that the SIO does not increase the coherence measure. Therefore, it becomes important to understand what kinds of errors actually occur in IBM Q processor.}

{  As addressed in  \cite{Chen2019}, the IBM Q processor has two different types of errors, (i) classical errors from the dcoherence represented by  the shrinkage of arrow in the Bloch-vector representation in figure \ref{Figure1}-(a), and (ii) nonclassical errors represented by the tilt of measurement directions. In the context of coherence, both errors are critical because coherence is not only sensitive to the decoherence but also basis-dependent. However, two types of errors could cause different effects on measurement coherence. For instance, the classical errors lead to uniform decrease of coherence regardless of measurement directions, while the nonclassical ones lead to different behaviors of coherence depending on measurement directions. Therefore, classical errors seem to be dominant in our experimental data, because the ratios of  experimental values to theoretical ones are rather uniform in the range from 0.9 to 0.95 regardless of measurement directions.  }
We thus attribute this trend to { classical} errors occurring in the machine, which could be further confirmed through other related experiments \cite{Chen2019}.

\section{Conclusion}

In summary, we have explored how to quantify coherence of measurement in a resource theoretical framework. A resource theory generally consists of two ingredients, free resources and free operations. In this work, we defined  a free resource as an incoherent measurement \eqref{IncM} by which one cannot give access to coherence of states experimentally. We also defined a free operation as a strictly incoherent operation (SIO) that cannot create a coherent measurement from an incoherent one. We remark that the SIOs are known as a physically well-motivated set of free operations for coherence, as one can neither create nor use coherence via a SIO \cite{Winter2016,Yadin2016}. In this framework, we showed that a statistical distance can be used to define a coherence monotone of measurements  \eqref{CohMonotone}. In particular, the relative entropy gives us a monotone that fulfills all postulates ($\mathcal C$1)-($\mathcal C$4). On the other hand, we introduced the $l_\infty$ norm-based coherence monotone $\mathcal C_{l_\infty}$ that is determined by off-diagonal elements of each POVM component. It can therefore be readily calculated without the convex optimization. Furthermore, we showed that $\mathcal C_{l_\infty}$ gives a lower bound on the robustness measure $\mathcal R_C$ in \eqref{Relation} and they become identical for two-outcome POVMs on $\mathcal H_2$.

To address the coherence of measurement experimentally, we need to have a full description of a given POVM. To this end, we introduced a procedure for the detector tomography by adopting the quantum process tomography \cite{Isaac1997}. In this procedure, $d^2$ linearly independent probe states are prepared and each probe state is measured by an unknown measurement. Then one can obtain its full description from the relation \eqref{DetectorT}. One of advantages from this procedure is that it can be implemented without the global reconstruction that may require a demanding numerical computation. Furthermore, as we prepare a particular set of probe states \eqref{States}, we obtain the off-diagonal elements of each POVM component straightforwardly and selectively. Finally,  we illustrated the feasibility of our approach by performing the single qubit experiment on the 20 qubit IBM Quantum system ''Singapore'' as shown in figure \ref{Figure1}.

We hope that our work could lead to further researches on quantitatively characterizing the role of coherence of measurements in operational tasks such as quantum measurement engine and measurement-based quantum computation. From the fundamental perspective, our approach may provide some insights to exploring other characteristics of quantum measurements such as incompatibility and separable measurements, which will be all subject to future studies.

\ack

KB, AS, JL and JK supported by KIAS Individual Grants
(CG074701, CG070301, CG073101 and CG014604) at Korea Institute for Advanced Study, respectively. We acknowledge the support of Samsung Advanced Institute of Technology for the use of IBM Quantum systems.

\appendix

\section{Strictly incoherent operation \cite{Winter2016,Yadin2016}}\label{Appdx1}

An incoherent Kraus operator is generally written as \cite{Chitambar2016b}
\begin{align*}
K_\mu=\sum_{i=0}^{d-1} c_{\mu,i} |f_\mu(i)\rangle\langle i|
\end{align*}
with coefficients satisfying the completeness relation
\begin{align*}
\sum_{\mu:f_\mu(i)=f_\mu(j)} c_{\mu,j}^* c_{\mu,i}=\delta_{ij}, 
\end{align*}
where $f_\mu$ is a function from $\{0,...,d-1\}$ to $\{0,...,d-1\}$. The conjugate of $K_\mu$ must also be written in this form to be a strictly incoherent operation. It implies $f_\mu$ should be an invertible function, i.e., 
permutation $\pi_\mu$ depending on $\mu$ \cite{Chitambar2016b}. In this case, the Kraus operator can be decomposed as
\begin{align*}
K_\mu=\sum_{i=0}^{d-1} c_{\mu,i} |\pi_\mu(i)\rangle\langle i|= \sum_{i=0}^{d-1} c_{\mu,i} {V}_\mu |i\rangle\langle i|= {V}_\mu\tilde{K}_\mu,
\end{align*}
where $V_\mu$ is an unitary operation corresponding to the permutation $\pi_\mu$ and  $\tilde{K}_\mu=\sum_{i=0}^{d-1} c_{\mu,i} |i\rangle\langle i|$ is a genuinely incoherent Kraus operator \cite{Yadin2016}.  Here, the coefficients satisfy $\sum_{\mu} |c_{\mu,i}|^2=1$ for the completeness relation.

\subsection{Dual of a selective operation}

In the condition ($\mathcal C$3), we have the measurement $\mathbf A'=\{K_\mu^\dagger A_a K_\mu\}$, which is given after applying the dual operation of a selective SIO to $\mathbf A$. To explain it more specifically, let us introduce an ancilla storing classical information about the selection $\mu$. We assume an operation $\mathcal E$ given by $\{K_\mu \otimes U_\mu^{cl}\}_{\mu=0}^{n_{\mathcal E}-1}$ acts on a quantum state $\rho\otimes |0\rangle\langle 0 |^{cl}$, where $K_\mu$ is a Kraus operator and $U_\mu^{cl}$ is an unitary operator shifting the basis $\{|i\rangle^{cl}\}_{i=0}^{n_{\mathcal E}-1}$ by $\mu$, i.e., $U_\mu^{cl}|i\rangle^{cl}=|i+\mu \text{ mod } n_{\mathcal E}\rangle^{cl}$ in $n_{\mathcal E}$-dimensional Hilbert space. By applying $\mathcal E$ to the state, we have
\begin{align*}
\mathcal E(\rho\otimes |0\rangle\langle 0|^{cl})=\sum_{\mu=0}^{n_{\mathcal E}-1} K_\mu \rho K_\mu^{\dagger} \otimes |\mu\rangle\langle \mu|^{cl}.
\end{align*}
Then, we can make a selective operation on the initial measurement by performing a measurement given by $\{A_a\otimes |i\rangle\langle i|^{cl}\}_{a\in \mathcal Z_n,i\in \mathcal Z_{n_{\mathcal E}}}$, where $\mathcal Z_n=\{0,1,...,n-1\}$. As a result, this measurement gives a distribution $p(m,\mu)=\Tr[K_\mu \rho K_\mu^{\dagger} A_a]$. This result is equivalent to what we would obtain by measuring $\rho$ via $\{K_\mu^{\dagger} A_a K_\mu\}_{a\in \mathcal Z_n,\mu\in \mathcal Z_{n_{\mathcal E}}}$. Therefore, we can consider $\mathbf A'$ as a measurement given after applying the dual operation of a selective SIO.

\section{Proofs of conditions for $\mathcal C_D(A)$ and $\mathcal C_S(A)$}\label{Appdx2}

We first prove that $\mathcal C_D(\mathbf A)$ satisfies ($\mathcal C$1) and ($\mathcal C$2) as follows.

\noindent{\it Proof of ($\mathcal C$1) for $\mathcal C_D(\mathbf A)$}. Any statistical distance is supposed to give $D(\mathbf A,\mathbf M)=0$ if and only if two probability distributions associated with POVMs $\mathbf A$ and $\mathbf M$ are identical. Otherwise it should be strictly positive. Thus, in the definition of $\mathcal C_D(\mathbf A)$, $\sup_\rho D(\mathbf A,\mathbf M)=0$ implies that two probability distributions  are identical for all states, i.e., $\Tr[\rho A_a]=\Tr[\rho M_a]$ for arbitrary $a$ and $\rho$. Equivalently, it can be said that $\mathbf A$ and $\mathbf M$ are identical. Thus, if $\mathbf A$ is an incoherent measurement, $\mathcal C_D(\mathbf A)$ vanishes during the minimization over all incoherent measurements. Otherwise, it gives a strict positive value. \hspace*{\fill}\qedsymbol

\noindent{\it Proof of ($\mathcal C$2) for $\mathcal C_D (\mathbf A)$.} If $\mathcal E$ is a SIO, then we have

\begin{align*}
\mathcal C_D (\mathbf A)&= \min_{\mathbf M\in \mathcal I(d,n)} \sup_\rho D_\rho(\mathbf A,\mathbf M)\\
&\geq \min_{\mathbf M\in \mathcal I(d,n)} \sup_{\mathcal{E}(\rho)} D_{\mathcal{E}(\rho)}(\mathbf A,\mathbf M) \\
&=\min_{\mathbf M\in \mathcal I(d,n)} \sup_\rho D_\rho(\mathcal E^*(\mathbf A),\mathcal E^*(\mathbf M))\\
&\geq \min_{\mathbf M\in \mathcal I(d,n)} \sup_\rho D_\rho(\mathcal E^*(\mathbf A),\mathbf M)=\mathcal C_D (\mathcal{E}^*( \mathbf A)).
\end{align*}
Here, the first inequality comes from the fact that $\{\mathcal E(\rho)| \text{ for all }\rho \}$ is a subset of $\{\rho\}$. Similarly, we have the second inequality due to $\{\mathcal E^*(\mathbf M)| \mathbf M\in\mathcal I(d,n) \}\subset \mathcal I(d,n)$. \hspace*{\fill}\qedsymbol

Furthermore, the convexity of a statistical distance $D$ implies the convexity of  $\mathcal C_D$, which is proved as follows.

\noindent{\it Proof of ($\mathcal C$4) for $\mathcal C_D(\mathbf A)$ with the convexity of $D$.} Let us assume the convexity of $D$. That is, for pairs of probability distributions $(p_{\mathbf A_k},p_{\mathbf M_k})$ associated with measurements $\mathbf A_k$ and $\mathbf M_k$, respectively, we have $D(\sum_k q_k \mathbf A_k, \sum_k q_k \mathbf M_k)\leq \sum_k q_k D(\mathbf A_k,\mathbf M_k)$ with the probability $q_k$. By using this property, we have
\begin{align*}
\mathcal C_D (\sum_k q_k \mathbf A_k)&= \min_{\mathbf M_k\in \mathcal I(d,n)} \sup_\rho D_\rho(\sum_k q_k \mathbf A_k,\sum_k q_k \mathbf M_k)\\
&\leq \min_{\mathbf M_k\in \mathcal I(d,n)} \sup_{\rho} \sum_k q_k D_{\rho}(\mathbf A_k,\mathbf M_k) \\
&\leq \sum_k q_k \min_{\mathbf M_k\in \mathcal I(d,n)}  \sup_{\rho} D_{\rho}(\mathbf A_k,\mathbf M_k)  \\
&= \sum_k q_k \mathcal C_D (\mathbf A_k).
\end{align*}
In the first line, one can write the minimization over $\sum_k q_k \mathbf M_k$ without loss of generality. 
Then, the first inequality comes from the convexity of $D$ and the second one is given by taking supremum for each distance. \hspace*{\fill}\qedsymbol

According to the above proofs, $\mathcal C_S$ holds the conditions ($\mathcal C$1-2) and ($\mathcal C$4) due to its positive definiteness and  convexity. Let us prove the satisfaction of ($\mathcal C$3) for $\mathcal C_S$ in the following.

\noindent{\it Proof of ($\mathcal C$3) for $\mathcal C_S (\mathbf A)$.} 
Let a selective SIO $\mathcal E$ be described by a set of Kraus operators $\{K_\mu\}_{\mu=0}^{n_{\mathcal E}-1}$. Then, we have the following relations for $\mathbf A\in \mathcal M(d,n)$.
\begin{align*}
&\min_{\mathbf M\in\mathcal I(d,n)} \sup_\rho S_{\rho}(\mathbf A\|\mathbf M)\\& \geq \min_{\mathbf M\in\mathcal I(d,n)} \sup_\rho\left( \sum_{\mu=0}^{n_\mathcal{E}-1} p_\mu S_{\rho_\mu}(\mathbf A\| \mathbf M) \right) \\
&=\min_{\mathbf M\in\mathcal I(d,n)} \sup_\rho\left( \sum_{\mu=0}^{n_\mathcal{E}-1} p_\mu \sum_{a=0}^{n-1} p_{\mathbf A}(a|\mu)\log \frac{p_{\mathbf A}(a|\mu)}{p_{\mathbf M}(a|\mu)} \right)\\
&=\min_{\mathbf M\in\mathcal I(d,n)} \sup_\rho\left( \sum_{\mu=0}^{n_\mathcal{E}-1} \sum_{a=0}^{n-1} p_\mu p_{\mathbf A}(a|\mu)\log \frac{p_\mu p_{\mathbf A}(a|\mu)}{p_\mu p_{\mathbf M}(a|\mu)} \right)\\
&=\min_{\mathbf M\in\mathcal I(d,n)} \sup_\rho\left( \sum_{\mu=0}^{n_\mathcal{E}-1} \sum_{a=0}^{n-1} \Tr[K_\mu \rho K_\mu^\dagger A_a] \log \frac{\Tr[K_\mu \rho K_\mu^\dagger A_a]}{\Tr[K_\mu \rho K_\mu^\dagger M_a]} \right)\\
&=\min_{\mathbf M\in\mathcal I(d,n)} \sup_\rho\left( \sum_{\mu=0}^{n_\mathcal{E}-1} \sum_{a=0}^{n-1} \Tr[ \rho K_\mu^\dagger A_aK_\mu] \log \frac{\Tr[\rho K_\mu^\dagger A_a K_\mu] }{\Tr[ \rho K_\mu^\dagger M_a K_\mu]} \right)\\
&=\min_{\mathbf M\in\mathcal I(d,n)} \sup_\rho S(\mathbf A'\|\mathbf M')\\
&\geq\min_{\mathbf M\in\mathcal I(d,n\cdot n_{\mathcal E} )} \sup_\rho S(\mathbf A'\|\mathbf M)
\end{align*}
where a conditional state $\rho_\mu=K_\mu \rho K_\mu^\dagger/p_\mu$ is obtained with the probability  $p_\mu = \Tr[K_\mu \rho K_\mu^\dagger]$ after applying $\mathcal E$ to $\rho$. Here, the first inequality comes from $\sup_\rho S_\rho (\mathbf A\| \mathbf M) \geq \sup_\rho S_{\rho_\mu} (\mathbf A\| \mathbf M)$ because $\{\rho_\mu| \forall \rho\} $ is a subset of 
density operators on $\mathcal H_d$. Before the second inequality, we denote the POVMs emerging after applying the dual of the seletive SIO $\mathcal E$ to $\mathbf M\in\mathcal I({d, n})$ by $\mathbf M' =\{K_\mu^\dagger M_a K_\mu\}_{a,\mu}\in\mathcal I({d, n\cdot n_{\mathcal E}})$. Because $\mathbf M'$ defines a subset of incoherent measurements in $\mathcal I({d, n\cdot n_{\mathcal E}})$, the minimization over all incoherent measurement on $\mathcal I(d,n\cdot n_{\mathcal E} )$ gives a smaller value. Note that this proof is applicable to other statistical distances if they satisfy $\sum_{\mu=0}^{n_{\mathcal E}-1} p_\mu D_{\rho_\mu}(\mathbf A\|\mathbf M)=D_{\rho}(\mathbf A'\|\mathbf M')$. \hspace*{\fill}\qedsymbol

\section{Proofs of conditions for $\mathcal C(\mathbf A)$}\label{Appdx3}

We prove that $\mathcal C_{l_\infty}(\mathbf A)$ satisfies requirements ($\mathcal C$1-4) as follows.

\noindent{\it Proof of ($\mathcal C$1) for $\mathcal C_{l_\infty}(\mathbf A)$.} 
If a POVM $\mathbf A \in \mathcal M(d,n)$ is incoherent, then $|\langle i|A_a| j \rangle|=0$ for all $i,j$ and $a$, and vice versa. \hspace*{\fill}\qedsymbol

\noindent{\it Proof of ($\mathcal C$2) for $\mathcal C_{l_\infty}(\mathbf A)$.} 
$\mathcal C_{l_\infty}(\mathbf A)$ is independent of phase factors. Thus, without loss of generality, we assume that Kraus operators of a SIO $\mathcal E$ are given as 
\begin{align}\label{KrausNoPhase}
K_\mu=\sum_{i=0}^{d-1} \sqrt{p_{\mu|i}}|\pi_\mu(i)\rangle\langle i|.
\end{align}
Then, applying the dual of $\mathcal E$ gives rise to 
$$\mathcal E^*(\mathbf A) = \left\{\sum_{\mu=0}^{n_{\mathcal E}-1}\sum_{i,j=0}^{d-1} \sqrt{p_{\mu|i}p_{\mu|j}} \langle \pi_\mu(i)|A_a|\pi_\mu(j)\rangle |i\rangle\langle j|  \right\}_{a=0}^{n-1}.$$
Its coherence monotone $\mathcal C$ becomes smaller as follows:
\begin{align*}
\mathcal C_{l_\infty}(\mathcal E^*(\mathbf A))&= \max_{i<j} \sum_{a=0}^{n-1} \left|\sum_{\mu=0}^{n_{\mathcal E}-1} \sqrt{p_{\mu|i}p_{\mu|j}} \langle \pi_\mu(i)|A_a|\pi_\mu(j)\rangle\right|\\
&\leq \max_{i<j} \sum_{\mu=0}^{n_{\mathcal E}-1} \sqrt{p_{\mu|i}p_{\mu|j}}  \sum_{a=0}^{n-1} \left|\langle \pi_\mu(i)|A_a|\pi_\mu(j)\rangle\right|\\
&\leq \mathcal C_{l_\infty}(\mathbf A) \max_{i<j} \sum_{\mu=0}^{n_{\mathcal E}-1} \sqrt{p_{\mu|i}p_{\mu|j}}\\
&\leq \mathcal C_{l_\infty}(\mathbf A).
\end{align*}
The first inequality is due to the subadditivity of the absolute value. We have the second inequality by the definition of $\mathcal C_{l_\infty}$. The third inequality comes out of the Cauchy–Schwarz inequality. 
\hspace*{\fill}\qedsymbol

\noindent{\it Proof of ($\mathcal C$3) for $\mathcal C_{l_\infty}(\mathbf A)$.} For the same reason as above, we again restrict ourselves without loss of generality to a selective SIO $\mathcal E$ described by Kraus operators in \eqref{KrausNoPhase}. Then applying it to $\mathbf A\in\mathcal M(d,n)$, we have the POVM $\mathbf A'=\{ K_\mu^\dagger A_a K_\mu \}\in \mathcal M(d,n\cdot n_{\mathcal E})$, where each POVM component is expanded as
$$K_\mu^\dagger A_a K_\mu= \sum_{i,j=0}^{d-1} \sqrt{p_{\mu|i}p_{\mu|j}} \langle \pi_\mu(i)|A_a|\pi_\mu(j)\rangle |i\rangle\langle j|.  $$
Its coherence monotone $\mathcal C_{l_\infty}$ is given as
\begin{align*}
\mathcal C_{l_\infty}(\mathbf A')&= \max_{i<j} \sum_{\mu=0}^{n_{\mathcal E}-1}\sum_{a=0}^{n-1} \left| \sqrt{p_{\mu|i}p_{\mu|j}} \langle \pi_\mu(i)|A_a|\pi_\mu(j)\rangle\right|.
\end{align*}
In the same way that we prove ($\mathcal C$2) for $\mathcal C_{l_\infty}$, the Cauchy-Schwarz inequality implies the monotonicity $\mathcal C_{l_\infty}(\mathbf A')\leq \mathcal C_{l_\infty}(\mathbf A).$
\hspace*{\fill}\qedsymbol

\noindent{\it Proof of ($\mathcal C$4) for $\mathcal C_{l_\infty}(\mathbf A)$.} 
We consider the convex combination of $\mathbf A, \mathbf B\in \mathcal M(d,n)$ that is the POVM $p \mathbf A+(1-p) \mathbf B=\{pA_a+(1-p)B_a\}_{a=0}^{n-1}$ with weight $p$. Then its coherence monotone $\mathcal C_{l_\infty}$ is given as
\begin{align*}
\mathcal C_{l_\infty}(p\mathbf A+(1-p)&\mathbf B) = \max_{i<j} \sum_{a=0}^{n-1} \left| \langle i|p A_a+(1-p)B_a|j\rangle\right|\\
\leq&  \max_{i<j} \sum_{a=0}^{n-1}\big(  p\left|\langle i|A_a|j\rangle|+(1-p)|\langle i|B_a|j\rangle\right|\big)\\
\leq& p \mathcal C_{l_\infty}(\mathbf A)+(1-p)\mathcal C_{l_\infty}(\mathbf B).
\end{align*}
Here, the first inequality comes from the subadditivity of the absolute value, and the second one by separately applying the maximization.
\hspace*{\fill}\qedsymbol

\section{A counter example of the monotonicity for $\mathcal C_{l_1}$}\label{Appdx4}

We show that $\mathcal C_{l_1}$ does not fulfill the monotonicity ($\mathcal C$2) by a counter example. Let us consider a POVM $\mathbf A=\{A_{\pm}\}\in \mathcal M(4,2)$, where POVM components are given by
\[A_\pm=\begin{pmatrix}
1/2 & \pm1/2 & 0 &0\\
\pm1/2 & 1/2 & 0 & 0\\
0 & 0& 1/2 & 0\\
0 & 0 & 0 &1/2
\end{pmatrix}.\]
Then, as we apply the dual of a SIO $\mathcal E$ described by the Kraus operators 
\[K_0=\begin{pmatrix}
1 & 0 & 0 &0\\
0 & 1 & 0 & 0\\
0 & 0& 0 & 0\\
0 & 0 & 0 &0
\end{pmatrix}, 
K_1=\begin{pmatrix}
0 & 0 & 1 &0\\
0 & 0 & 0 & 1\\
0 & 0& 0 & 0\\
0 & 0 & 0 &0
\end{pmatrix},\]
the POVM components become
\[\mathcal E^*(A_\pm)=\begin{pmatrix}
1/2 & \pm1/2 & 0 &0\\
\pm1/2 & 1/2 & 0 & 0\\
0 & 0& 1/2 & \pm 1/2\\
0 & 0 & \pm 1/2 &1/2
\end{pmatrix},\]
respectively. Consequently, one can easily identify that $\mathcal C_{l_1}(\mathbf A)=2\leq \mathcal C_{l_1}(\mathcal E^*(\mathbf A))=4$.

\section{Proof of Theorem 1}\label{Prf of Thm1}

To prove Theorem 1, we need the following lemma. 

\begin{lemma}
For any state $\rho$ on $\mathcal H_d$, its off-diagonal elements satisfy
\begin{align}\label{lemma1}
|\langle i|\rho|j\rangle|\leq \frac{1}{2} \text{ for } i\neq j.
\end{align}
\end{lemma}
\begin{proof}
A state $\rho$ is positive semidefinite if and only if the determinant of every submatrix indexed by the same rows and columns of $\rho$, i.e., every principal submatrix of $\rho$, is nonnegative \cite{Zhang2011}. Thus, the determinant of any $2\times 2$ submatrix $\rho(i,j)$ of $\rho$ given by 
\[\rho(i,j)=\begin{pmatrix}
\rho_{ii}& \rho_{ij}\\
\rho_{ji} & \rho_{jj}
\end{pmatrix},\]
must be non-negative, namely $\sqrt{\rho_{ii}\rho_{jj}}\geq|\rho_{ji}|$, where $\rho_{ij}\equiv\langle i |\rho|j\rangle$. The Cauchy-Schwarz inequality with $\rho_{ii}+\rho_{jj}\leq 1$ and $\rho_{ii}\geq0$ for any $i,j$ then implies $|\rho_{ji}|\leq 1/2$.
\end{proof}

\noindent{\it Proof of Theorem 1.}

{
The robustness is efficiently cast by using semidefinite programming(SDP), and particularly in our case its dual problem is written as the following SDP \cite{Oszmaniec2019}.
\begin{align}\label{RobSDP}
 maximize &\;\;\;\;\; \sum_{a=0}^{n-1} \Tr[Z_a A_a ]-1 \\
 subject\;to &\;\;\;\;\; \forall i, a\;\;\; Z_a\geq 0, \;\; \langle i|Z_a|i\rangle=\langle i|Z_n|i\rangle \label{RobSDP1} \\
 &\;\;\;\;\; \Tr[Z_n]=1. \label{RobSDP2}
\end{align}
}

To prove the lower bound, let us rewrite the maximized function \eqref{RobSDP} as
\begin{align}\label{PrfTh1}
 &\sum_{a=0}^{n-1} \Tr[Z_a A_a]-1=  \sum_{a=0}^{n-1} \sum_{i,j=0}^{d-1} \langle i |Z_a|j\rangle \langle j | A_a|i\rangle -1\nonumber\\
&= \sum_{i=0}^{d-1} \sum_{a=0}^{n-1}  \langle i |Z_a|i\rangle \langle i | A_a|i\rangle+ \sum_{i\neq j}^{d-1}\sum_{a=0}^{n-1} \langle i |Z_a|j\rangle \langle j | A_a|i\rangle -1\nonumber\\
&=\sum_{i\neq j}^{d-1} \sum_{a=0}^{n-1} \langle i |Z_a|j\rangle \langle j | A_a|i\rangle.
\end{align}
Here, we have the third equality by using the completeness of $\mathbf A$ with the fact that $\langle i|Z_a|i\rangle$ is independent of $a$ as specified in Eq. \eqref{RobSDP1}. 

Now, consider a set of states $\{|k,l,\theta_{a}\rangle\}_{a=0}^{n-1}$ given by 
\begin{align}\label{PrfTh1_2}
|k,l,\theta_a\rangle= \frac{1}{\sqrt{2}} (|k\rangle+e^{i\theta_a}|l\rangle),
\end{align}
which defines $Z_a\equiv|k,l,\theta_{a}\rangle\langle k,l,\theta_{a}|$ to be used in Eq. \eqref{PrfTh1}.
The last term in Eq. \eqref{PrfTh1} is then given by 
\begin{align}\label{E4}
\sum_{i\neq j}^{d-1} \sum_{a=0}^{n-1} \langle i |Z_a|j\rangle \langle j | A_a|i\rangle=\sum_{a=0}^{n-1} {Re}\left[e^{i\theta_a} \langle l |A_a| k \rangle\right] .
\end{align}
Let $\mathcal C_{l_\infty}(\mathbf A)=\max_{i,j}\sum_{a=0}^{n-1} | \langle i |A_a| j \rangle |$ be achieved by $\{i,j\}=\{k,l\}$. Setting  $\theta_a\equiv-\arg[\langle l |A_a| k \rangle]$, we see that \eqref{E4} becomes $\mathcal C_{l_\infty}(\mathbf A)$. 
This proves $\mathcal C_{l_\infty}(\mathbf A)\leq \mathcal R_C(\mathbf A)$, as $\mathcal C_{l_\infty}(\mathbf A)$ can always be achieved by choosing a certain set of states $Z_a$ in \eqref{PrfTh1_2}.

It is straightforward to prove the upper bound, as we apply the result of lemma 1 to \eqref{PrfTh1} as 
\begin{align*}
&\sum_{i\neq j}^{d-1} \sum_{a=0}^{n-1} \langle i |Z_a|j\rangle \langle j | A_a|i\rangle\\
&\leq \sum_{i\neq j}^{d-1} \sum_{a=0}^{n-1} |\langle i |Z_a|j\rangle| |\langle j | A_a|i\rangle|\\
&\leq \frac{1}{2}\sum_{i\neq j}^{d-1} \sum_{a=0}^{n-1} |\langle j | A_a|i\rangle|=\frac{1}{2} \mathcal C_{l_1}(\mathbf A).
\end{align*}
This completes the proof. \hspace*{\fill} \qedsymbol

{
\section{A POVM $\mathbf M \in \mathcal M(3,2)$ and amplitude damping channel}\label{Ex}

To see relationships among $\mathcal R_{\mathcal C}$, $\mathcal C_{l_\infty}$ and $\mathcal C_{l_1}$ for higher dimensions, we consider a dichotomic POVM $\mathbf G=\{G_0, G_1=I_3-G_0\}$ as an example in a 3-dimensional system, where the POVM component is defined as 
\[G_0=\begin{pmatrix}
0.528 & 0.263 & 0.042\\
0.263 & 0.137 & 0.026\\
0.042 & 0.026 & 0.008
\end{pmatrix}.\]
We assume that the amplitude damping channel is applied to the POVM $\mathbf M$, of which Kraus operators are given as
\begin{align*}
K_{\mu}^{AD}= \sum_{i=\mu}^{2} \sqrt{i \choose \mu} \sqrt{(1-\gamma)^{i-\mu}\gamma^{\mu}} |i-\mu\rangle\langle i|
\end{align*}
for $\mu=0,1,2$, where $\gamma$ is the amplitude damping rate. One can identify that the amplitude damping channel is a SIO as the Kraus operators coincide with the form of Kraus operators describing a SIO \eqref{SIO}.
In figure \ref{Comparing}, we plot $\mathcal R_{\mathcal C}$, $\mathcal C_{l_\infty}$ and $\mathcal C_{l_1}$ of $\mathbf G'=\{ \sum_{\mu=0}^2 K_\mu^{AD\dagger} G_g K_\mu^{AD} \}_{g=0,1}$ versus the damping rate $\gamma$.
}

\section{Experimental data from the IBM Q processor}\label{App_Exp}

\begin{figure}[t]
  \centering
  \includegraphics[width =7cm]{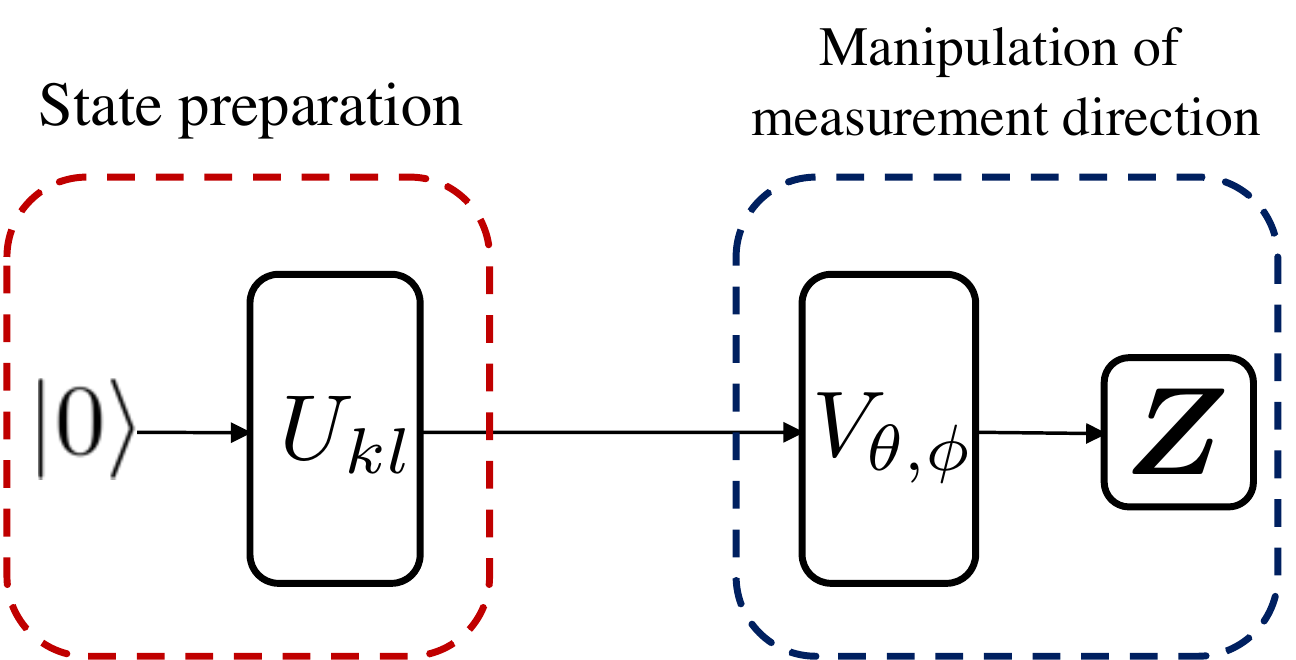}
  \caption{State preparation and manipulation of measurement direction on IBM Q processor.}\label{Scheme}
\end{figure}

\begin{figure*}[t]
  \centering
  \includegraphics[width =4cm]{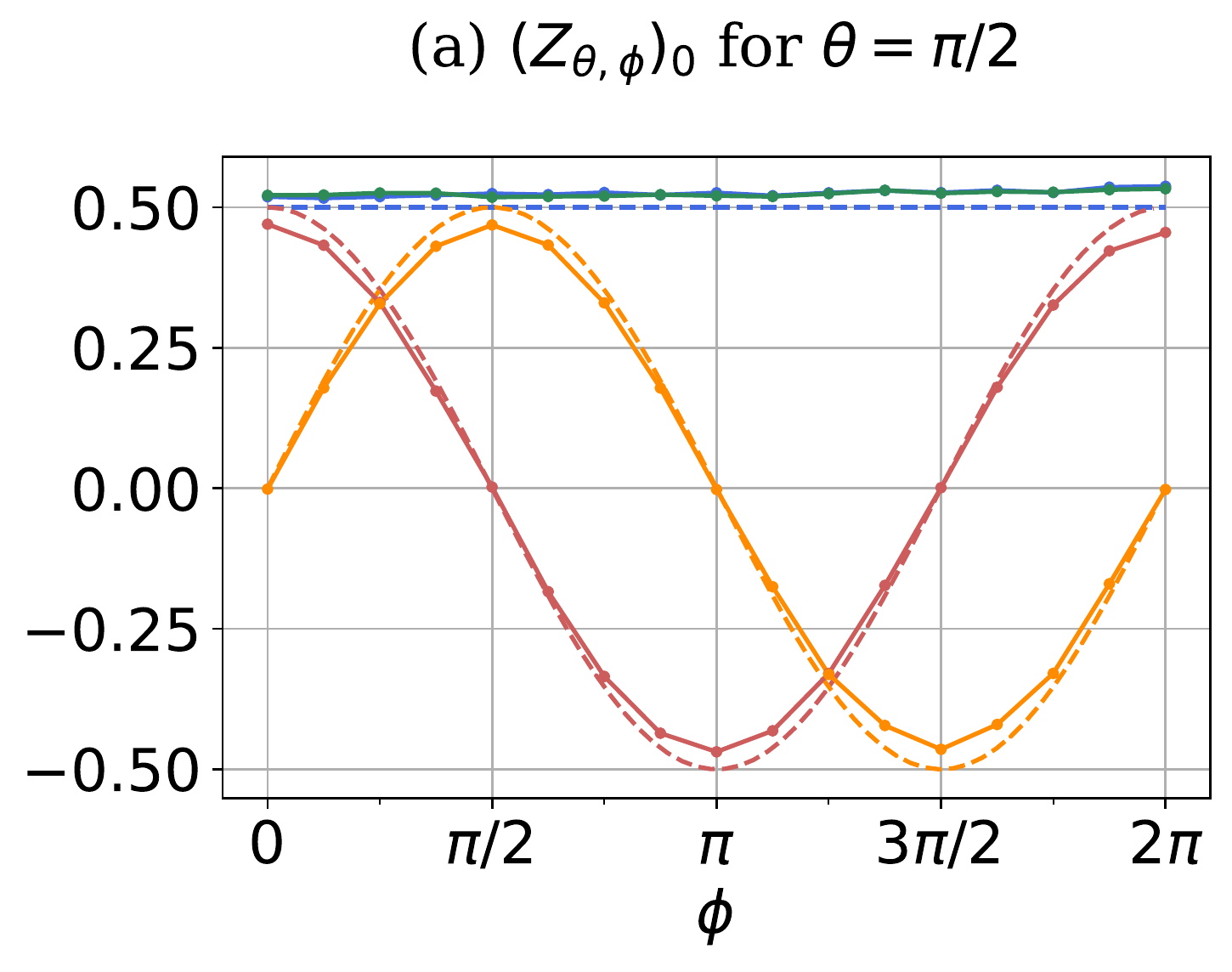}
  \includegraphics[width =4cm]{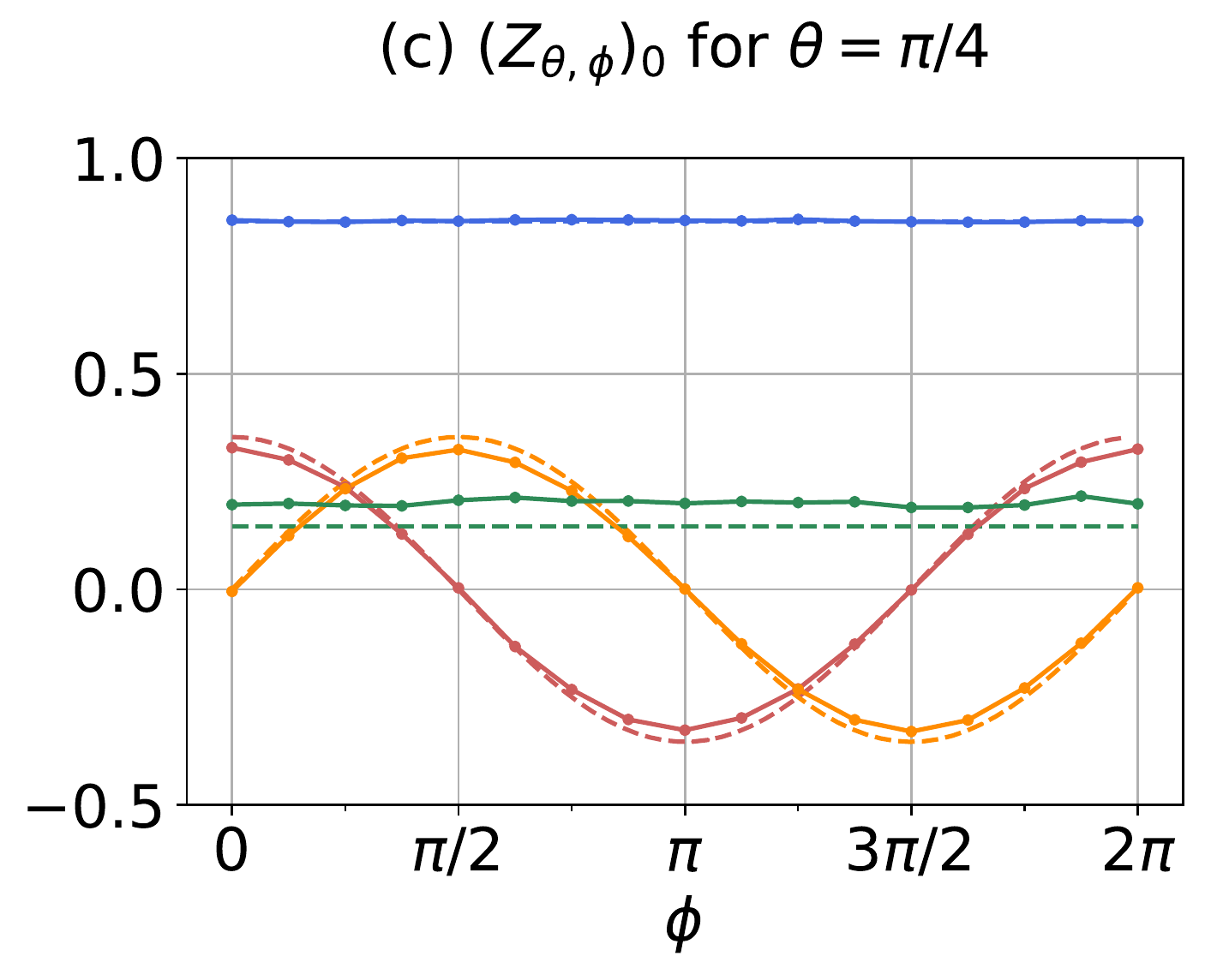}
  \includegraphics[width =6.3cm]{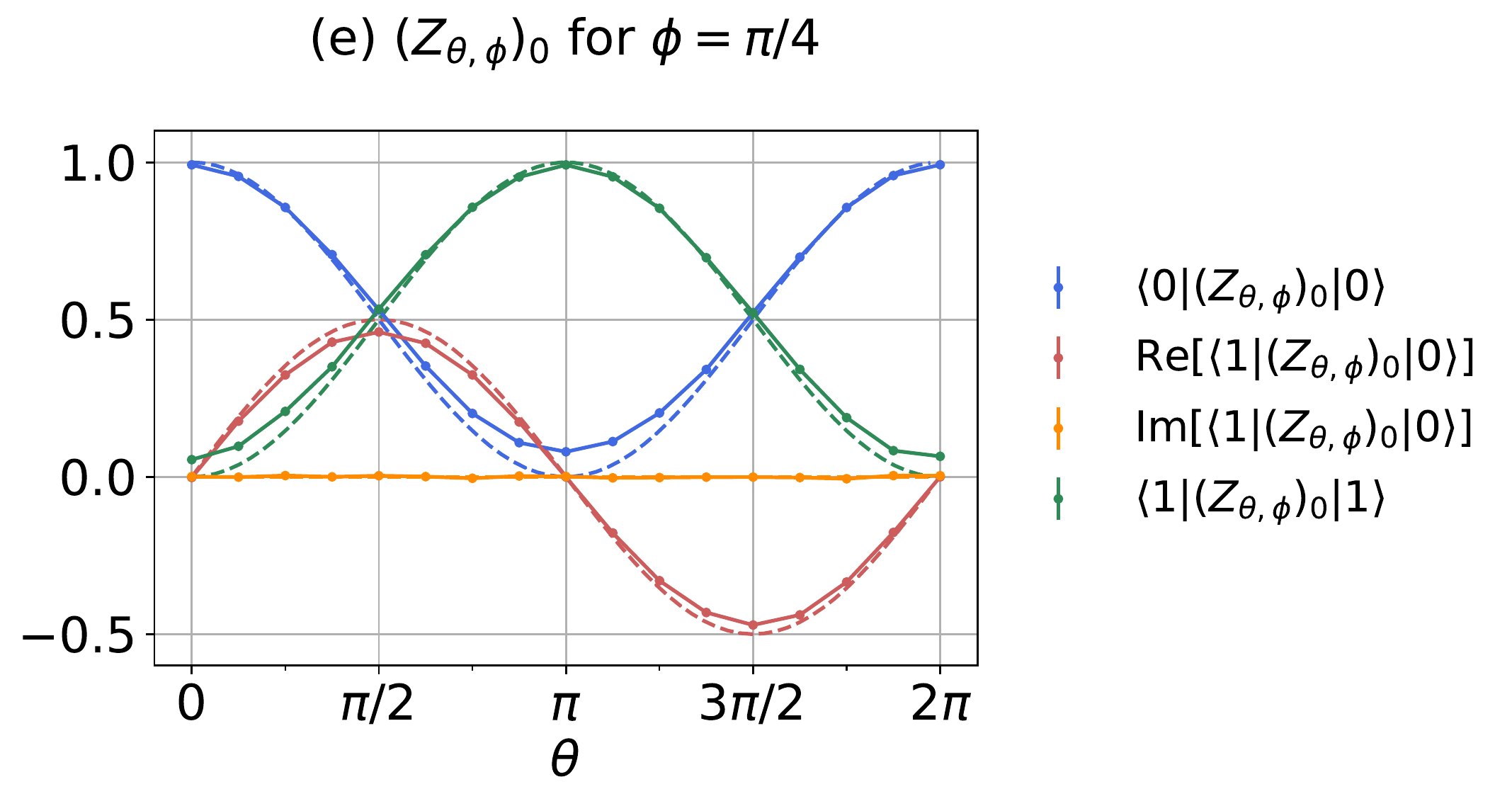}
  
    \includegraphics[width =4cm]{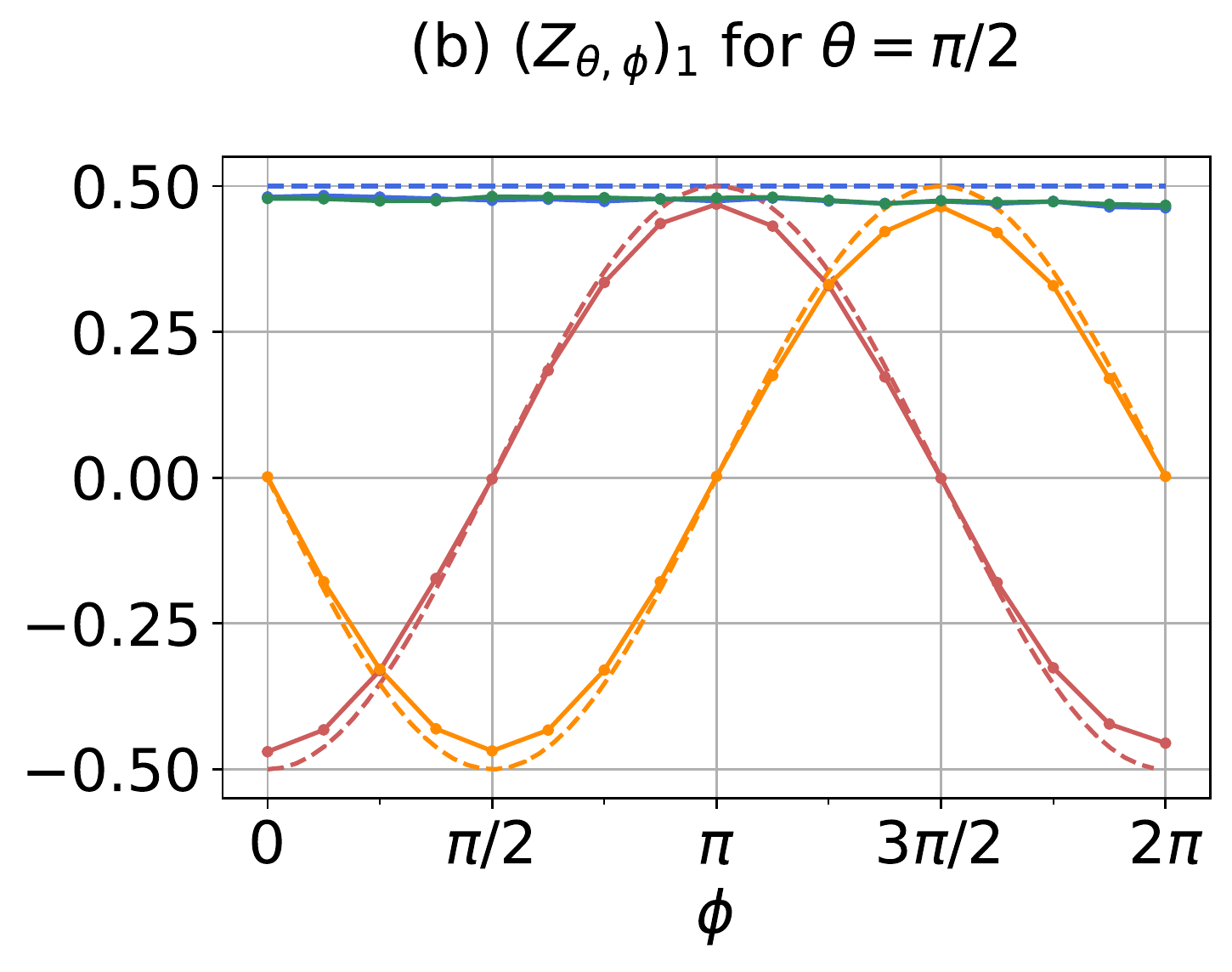}
  \includegraphics[width =4cm]{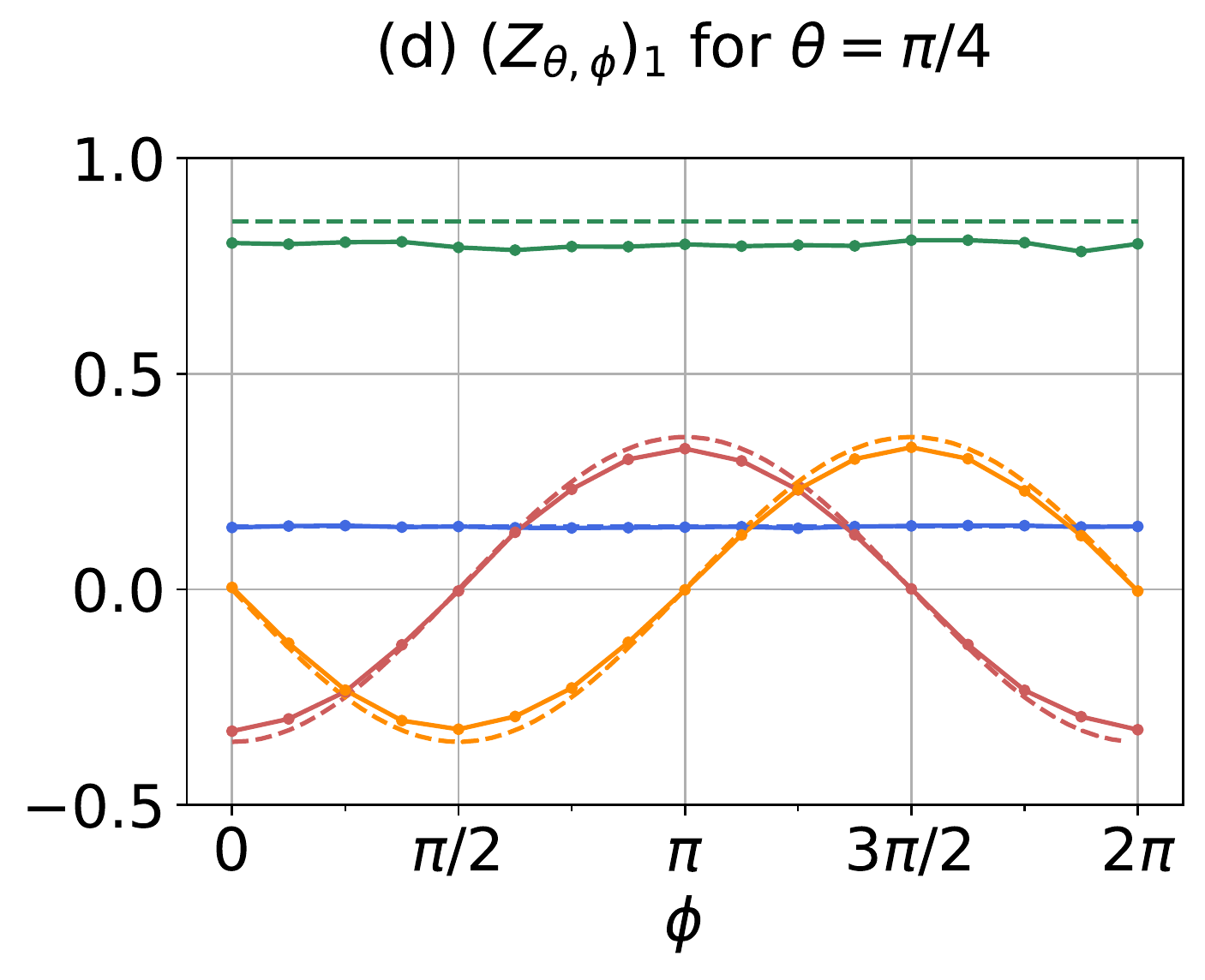}
  \includegraphics[width =6.3cm]{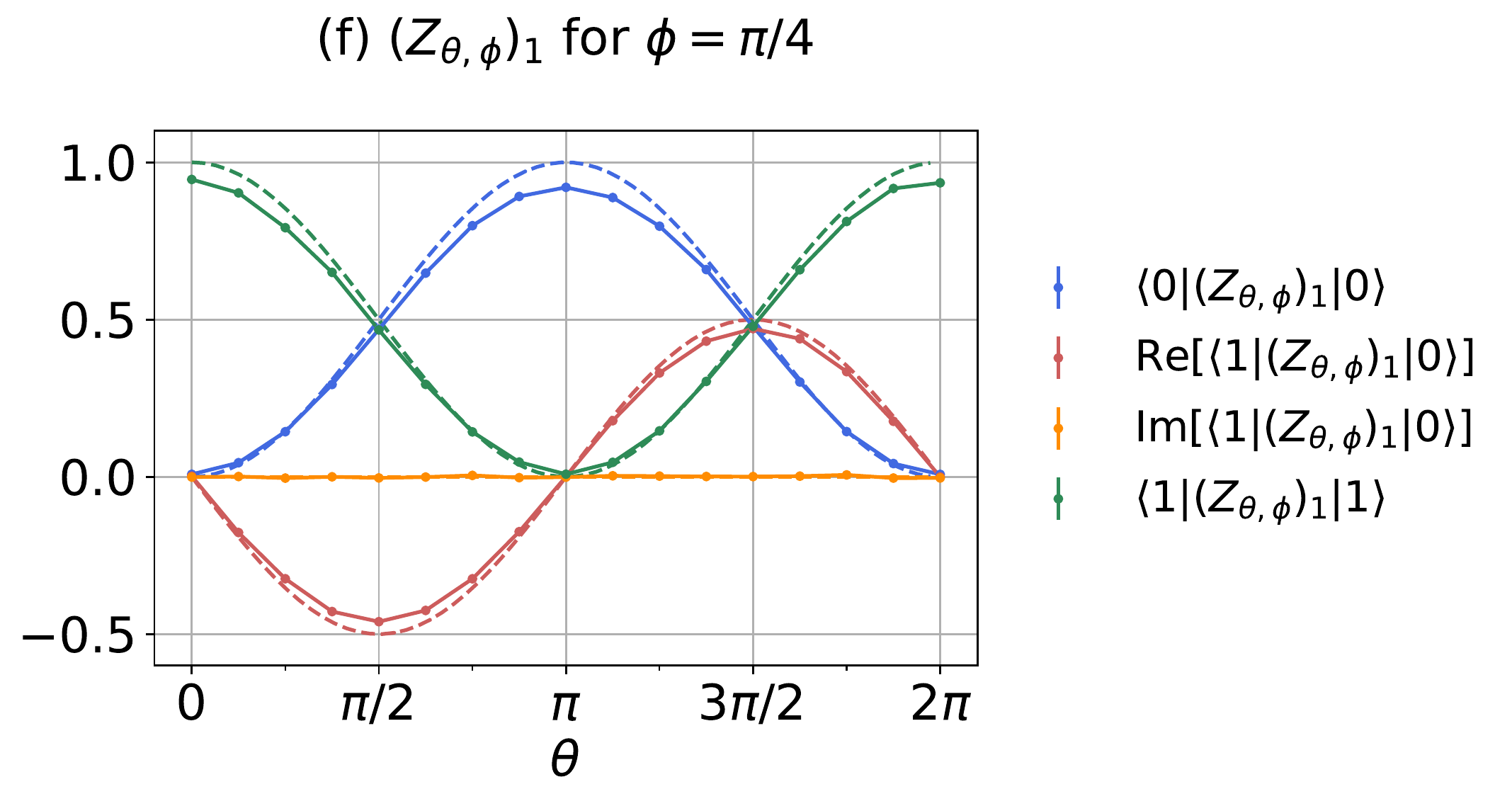}
  \caption{Matrix elements of $(Z_{\theta,\phi})_i=V_{\theta,\phi}^\dagger |i\rangle\langle i|V_{\theta,\phi}$ for $i=0,1$.}\label{Graphs}
\end{figure*}

In the IBM Q processor, the initial state is prepared in $|0\rangle$. Applying single-qubit gates $U_{00}=I$, $U_{01}=H$, $U_{10}=PH$ and $U_{11}=X$ to the initial state, we prepare the family of states $\{|\psi_{kl}\rangle=U_{kl}|0\rangle\}_{k,l=0,1}$, where  $H=(1/2)(|0\rangle\langle0|+|0\rangle\langle1|+|1\rangle\langle0|-|1\rangle\langle1|)$, $P=|0\rangle\langle0|+i|1\rangle\langle1|$ and $X=|0\rangle\langle1|+|1\rangle\langle0|$ are the Hadamard, the phase and the Pauli-$X$ gate, respectively. For each case, we apply the single-qubit gate $V_{\theta,\phi}$ before we execute the measurement $\bm Z$, as illustrated in figure \ref{Scheme}. Namely, $U_{kl}$ and $V_{\theta,\phi}$ are applied sequentially, 
with $U_{kl}$ used to prepare $|\psi_{kl}\rangle$ and $V_{\theta,\phi}$ to manipulate the measurement direction.

For a pair of angles $\theta$ and $\phi$, we obtain experimental data for each $k,l$ from the IBM Q processor with the setting of 8192 shots, and calculate the full description of measurement $\bm Z_{\theta,\phi}$ from the data. We repeat this procedure 10 times. In figure \ref{Graphs}, we plot the average values of matrix elements of POVM components $(Z_{\theta,\phi})_i=V_{\theta,\phi}^\dagger |i\rangle\langle i|V_{\theta,\phi}$ for $i=0,1$. From these values, we assess coherence of measurements in figure \ref{Figure1} of main text. We also note that the size of standard deviations is less than $10^{-3}$.

\section*{References}

\bibliographystyle{unsrt.bst}
\bibliography{mybibfile}

\end{document}